\documentclass[a4paper]{article}
\usepackage[margin=1in]{geometry}
\usepackage[dvipdfmx]{graphicx}
\usepackage{amsmath,amssymb}
\usepackage{amsthm}

\newtheorem{theorem}{Theorem}

\newtheorem{lemma}[theorem]{Lemma}
\newtheorem{corollary}[theorem]{Corollary}

\newcommand{\numocc}{\#\mathit{occ}}
\newcommand{\SUS}{\mathsf{SUS}}

\newcommand{\LS}{\mathcal{LS}}
\newcommand{\MS}{\mathcal{MS}}
\newcommand{\RS}{\mathcal{RS}}

\title{Tight bounds on the maximum number of shortest unique substrings}

\author{
  Takuya Mieno
  \quad
  Shunsuke Inenaga
  \quad
  Hideo Bannai
  \quad
  Masayuki Takeda\\
  {Department of Informatics, Kyushu University}\\
  {\texttt{\{takuya.mieno,inenaga,bannai,takeda\}@inf.kyushu-u.ac.jp}}
}

\date{}

\begin{document}

\maketitle

\begin{abstract}
  A substring $Q$ of a string $S$ is called
  a shortest unique substring (SUS) for interval $[s,t]$ in $S$,
  if $Q$ occurs exactly once in $S$,
  this occurrence of $Q$ contains interval $[s,t]$,
  and every substring of $S$ which contains interval $[s,t]$ 
  and is shorter than $Q$ occurs at least twice in $S$.
  The SUS problem is, given a string $S$,
  to preprocess $S$ so that for any subsequent query interval $[s,t]$
  all the SUSs for interval $[s,t]$ can be answered quickly.
  When $s = t$, 
  we call the SUSs for $[s, t]$ as \emph{point SUSs},
  and when $s \leq t$, 
  we call the SUSs for $[s, t]$ as \emph{interval SUSs}.
  There exist optimal $O(n)$-time preprocessing scheme which answers
  queries in optimal $O(k)$ time for both point and interval SUSs,
  where $n$ is the length of $S$ and $k$ is the number of 
  outputs for a given query.
  In this paper, we reveal structural, combinatorial properties
  underlying the SUS problem:
  Namely, we show that the number of intervals in $S$ that correspond to
  point SUSs for all query positions in $S$ is less than $1.5n$, 
  and show that this is a matching upper and lower bound.
  Also, we consider the maximum number of intervals in $S$
  that correspond to interval SUSs for all query intervals in $S$.
\end{abstract}

\section{Introduction} \label{sec:intro}
\subsection{Shortest unique substring (SUS) problems} \label{subsec:intro}

A substring $Q$ of a string $S$ is called
a \emph{shortest unique substring} (\emph{SUS}) for 
interval $[s, t]$ in $S$, 
if 
(1) $Q$ occurs exactly once in $S$,
(2) this occurrence of $Q$ contains 
interval $[s,t]$, and 
(3) every substring of $S$ which contains interval $[s,t]$ and is shorter than 
    $Q$ occurs at least twice in $S$.
The \emph{SUS problem} is to preprocess a given string $S$
so that for any subsequent query 
interval $[s,t]$, SUSs for interval $[s,t]$ can be answered quickly.
When $s = t$,
a query $[s, t]$ refers to a single position in the string $S$,
and the problem is specifically called the \emph{point SUS problem}.
For clarity, when $s \leq t$,
the problem is called the \emph{interval SUS problem}.

Pei et al.~\cite{Pei} were the first to consider 
the point SUS problem, motivated by some applications in bioinformatics.
They considered two versions of this problem,
depending on whether a single point 
SUS has to be returned (the \emph{single point SUS problem}) or
all point SUSs have to be returned (the \emph{all point SUSs problem})
for a query position.

There is a series of research for the single point SUS problem.
Pei et al.~\cite{Pei} gave an $O(n^2)$-time 
preprocessing scheme which returns a single 
point SUS for a query position in $O(1)$ time, 
where $n$ is the length of the input string.
Tsuruta et al.~\cite{Tsuruta} and 
Ileri et al.~\cite{ileri14:_short_unique_subst_query_revis}
independently showed optimal $O(n)$-time preprocessing schemes 
which return a single 
point SUS for a query position in $O(1)$ time.
Hon et al.~\cite{HonTX15} proposed an \emph{in-place} algorithm
for the same version of the problem,
achieving the same bounds as the above solutions.

For the all point SUS problem which is more difficult,
Tsuruta et al.~\cite{Tsuruta} and 
Ileri et al.~\cite{ileri14:_short_unique_subst_query_revis}
also showed optimal algorithms 
achieving $O(n)$ preprocessing time and $O(k)$ query time,
where $k$ is the number of all point SUSs for a query point.

Hu et al.~\cite{Hu} were the first to consider
the interval SUS problem, and they
proposed an optimal algorithm for the the interval SUS problem, 
using $O(n)$ time for preprocessing and $O(k')$ time for queries,
where $k'$ is the number of interval SUSs for a query interval.
Recently, Mieno et al.~\cite{MienoIBT16} proposed an algorithm
which solves the interval SUS problem on strings 
represented by \emph{run-length encoding} (RLE).
If $r$ is the size of the RLE of a given string of length $n$,
then $r \leq n$ always holds.
Mieno et al.'s algorithm uses $O(r)$ space,
requires $O(r \log r)$ time to construct,
and answers all SUSs for a query interval
in $O(k' + \sqrt{\log r / \log\log r})$ time.

A substring $X$ of a string $S$ is said to be a 
\emph{minimal unique substring} (\emph{MUS}) of $S$, if
(i) $X$ occurs in $S$ exactly once and 
(ii) every proper substring of $X$ occurs at least twice in $S$.
All the above algorithms for the SUS problems
pre-compute all MUSs of the input string $S$
(or some data structure which is essentially equivalent to MUSs),
and extensively use MUSs to return the SUSs for a query position or interval.

Tsuruta et al.~\cite{Tsuruta} showed that the maximum number of 
MUSs contained in a string of length $n$ is at most $n$.
This immediately follows from the fact that MUSs do not nest.
Mieno et al.~\cite{MienoIBT16} proved that 
the maximum number of MUSs in a string is bounded by $2r-1$,
where $r$ is the size of the RLE of the string.
They also showed a series of strings which have $2r-1$ MUSs,
and hence this bound is tight.
These properties played significant roles in designing
efficient algorithms for the SUS problems.

On the other hand, structural properties of SUSs are not well understood.
A trivial upperbound for the maximum number of intervals
that correspond to point SUSs is $3n$, 
since every MUS can be a SUS for some position of the input string $S$,
and for each query position $p$~($1 \leq p \leq n$),
there can be at most $2$ SUSs that are not MUSs
(one that ends at position $p$ and the other that begins at position $p$).

\subsection{Our contribution}

The main contribution of this paper is matching upper and lower bounds 
for the maximum number of SUSs for the point SUS problem, 
which translate to ``less than $1.5n$ point SUSs''. 
Namely, we prove that any string of length $n$ contains at most $(3n-1)/2$ SUSs
for the point SUS problem.
We give a series of strings which contains 
$(3n-1)/2$ SUSs for any odd number $n \geq 5$.
Therefore, our bound is tight,
and to our knowledge, this is the first non-trivial result
for structural properties of SUSs.

We also consider the maximum number of SUSs for the interval SUS problem.
In so doing, we exclude a special case where
a query interval $[s,t]$ itself is a unique substring
that occurs exactly once in $S$.
This is because we have $\Theta(n^2)$ bounds for such trivial SUSs.
We then prove that any string of length $n$ contains less than $2n$
\emph{non-trivial} SUSs for the interval SUS problem. 
We also prove that there exists 
a string of length $n$ which contains 
$(2-\varepsilon)n$ non-trivial SUSs for any small number $\varepsilon > 0$.

\subsection{Related work}

Xu~\cite{LR} introduced the \emph{longest repeat} (\emph{LR}) problem.
An interval $[i,j]$ of a string $S$ is said to be an LR 
for interval $[s, t]$ if 
(a) the substring $R = S[i..j]$ occurs at least twice in $S$,
(b) the occurrence $[i,j]$ of $R$ contains $[s, t]$ and 
(c) there does not exist an interval $[i', j']$ of $S$
    such that $j'-i' > j-i$,
    the substring $S[i'..j']$ occurs at least twice in $S$,
    and the interval $[i', j']$ contains interval $[s, t]$.        
The point and interval LR problems are defined analogously
as the point and interval SUS problems, respectively.
    
Xu~\cite{LR} presented an optimal algorithm
which, after $O(n)$-time preprocessing,
returns all LRs for a given interval 
in $O(k'')$ time, where $k''$ is the number of output LRs.
He claimed that although the point/interval SUS problems and
the point/interval LR problems look alike,
these problems are actually quite different,
with a support from an example where
an SUS and LR for the same query point seem rather unrelated.

Our $(3n-1)/2$ bound for the 
maximum number of SUSs for the point SUS problem
also supports his claim in the following sense: 
In the preprocessing, Xu's algorithm computes the set of 
\emph{maximal repeats} (\emph{MR}).
An interval $[i, j]$ of a string $S$ is said to be an MR
if (A) the substring $W = S[i..j]$ occurs at least twice in $S$,
and (B) for any $1 \leq i' \leq i \leq j \leq j' \leq n$ with $j'-i' > j-i$, 
every superstring $Y = S[i'..j']$ of $W$ occurs once in $S$.
It is easy to see that the maximum number of MRs is bounded by $n$,
since for any position in $S$,
there can be at most one MR that begins at that position.
This bound is also tight:
any even palindrome consisting of $n/2$ distinct characters
contains $n$ intervals for which the corresponding
substrings are MRs
(e.g., for even palindrome $\mathtt{abcdeedcba}$ of length $10$,
any interval $[i,i]$ for $1 \leq i \leq 10$ is an MR).
By definition, any LR of string $S$ is also an MR of $S$.
Hence, the maximum number of LRs is also bounded by $n$.
Since the above lower bound for MRs with palindromes also applies to LRs,
this upper bound for LRs is also tight.
Thus, there is a gap of $(n-1)/2$ between the maximum numbers of SUSs and LRs.

\section{Preliminaries} \label{sec:pre}

\subsection{Notations} \label{subsec:notation}
Let $\Sigma$ be the alphabet.
An element of $\Sigma^*$ is called a string.
We denote the length of string $S$ by $|S|$.
The empty string is the string of length $0$.
For any string $S$ of length $n$ and integer $1 \leq i \leq n$,
let $S[i]$ denote the $i$th character of $S$.
For any $1 \leq i \leq j \leq n$,
let $S[i..j]$ denote the substring of $S$ that starts
at position $i$ and ends at position $j$ in $S$.
For convenience, $S[i..j]$ is the empty string if $i > j$.
For any strings $S$ and $w$,
let $\numocc_S(w)$ denote the number of occurrences of $w$ in $S$,
namely, $\numocc_S(w) = |\{i : S[i..i+|w|-1] = w\}|$.

\subsection{MUSs and SUSs} \label{subsec:MUSandSUS}
Let $S$ be any string of length $n$,
and $w$ be any non-empty substring of $S$.
We say that $w$ is a \emph{repeating substring} of $S$
iff $\numocc_S(w) \geq 2$,
and that $w$ is a \emph{unique substring} of $S$
iff $\numocc_S(w) = 1$.
Since any unique substring $w$ of $S$ occurs exactly once in $S$,
we will sometimes identify $w$ with its corresponding interval $[i, j]$
such that $w = S[i..j]$.
We also say that interval $[i,j]$ is unique
iff the corresponding $S[i..j]$ is a unique substring of $S$.

A unique substring $w = S[i..j]$ of $S$ is said to be
a \emph{minimal unique substring} (\emph{MUS})
iff any proper substring of $w$ is a repeating substring,
namely, $\numocc_S(S[i'..j']) \geq 2$ for any $i'$ and $j'$
with $i' \geq i$, $j' \leq j$, and $j'-i' < j-i$.
Let $\mathcal{M}_S$ be the set of all MUSs in $S$,
namely, $\mathcal{M}_S = \{[i,j] : \mbox{$S[i..j]$ is a MUS of $S$}\}$.
The next lemma follows from the definition of MUSs.

\begin{lemma}[\cite{Tsuruta}]\label{lem:sizeof_mus}
  No element of $\mathcal{M}_S$ is nested in another
  element of $\mathcal{M}_S$, namely,
  any two MUSs $[i,j], [k, \ell] \in \mathcal{M}_S$ satisfy
  $[i,j] \not \subset [k, \ell]$ and $[k, \ell] \not \subset [i, j]$.
  Therefore, $0 < |\mathcal{M}_S| \leq n$.
\end{lemma}

For any substring $S[i..j]$ and an interval $[s,t]$ in $S$,
$S[i..j]$ is said to be a \emph{shortest unique substring} (\emph{SUS})
for interval $[s,t]$ iff
(1) $S[i..j]$ is a unique substring of $S$,
(2) $[s, t] \subset [i, j]$, and
(3) $S[i'..j']$ is a repeating substring of $S$
    for any $i',j'$ with $[s,t] \subset [i', j']$
    and $j'-i' < j-i$.
In particular, for any interval $[p,p]$ of length $1$ in $S$,
$S[i..j]$ is said to be a SUS containing position $p$.
We say that an SUS for some position is a \emph{point} SUS, 
and a SUS for some interval (including those of length $1$) 
is an \emph{interval} SUS.

Clearly, if $[s,t]$ is unique,
then $[s,t]$ is the only SUS for the interval $[s,t]$.
For any interval $[s, t] \subset [1, |S|]$, 
if $s \neq t$ and $[s,t]$ is unique, 
we say that $[s, t]$ is a \emph{trivial} SUS. 
Also, we say that $[s,t]$ is a \emph{non-trivial} SUS 
if $[s, t]$ is not a trivial SUS.
For any interval $[s,t] \subset [1, |S|]$,
let $\SUS_S([s,t])$ denote the set of interval SUSs of $S$ 
that contain query interval $[s,t]$, 
and $\mathcal{IS}_S$ the set of all non-trivial SUSs of $S$.
Also, for any position $p \in [1, |S|]$,
let $\SUS_S(p)$ denote the set of 
point SUSs of $S$ that contain query position $p$, 
and $\mathcal{PS}_S$ the set of all point SUSs of $S$,
namely, $\mathcal{PS}_S = \bigcup^{n}_{p=1} \SUS_S(p)$.
Fig.~\ref{fig:MUSandSUS}
shows examples of MUSs and SUSs.

\begin{figure}[tbp] 
 \centerline{
 \includegraphics[width=60mm]{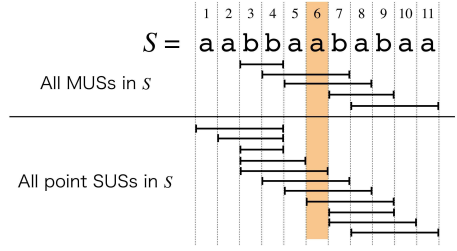}
 }
 \caption{
   For string $S = \mathtt{aabbaababaa}$,
   the set $\mathcal{M}_S = \{[3..4], [4..7], [5..8], [7..9], [8..11]\} = 
   \{\mathtt{bb}, \mathtt{baab}, \mathtt{aaba}, \mathtt{bab}, \mathtt{abaa}\}$
   of all MUSs of $S$ is shown in the upper part of the diagram.
   The set $\mathcal{PS}_S$ of all SUSs for all positions of string $S$
   is shown in the lower part of the diagram.
   For example, the intervals 
   $[3..6] = \mathtt{bbaa}$, 
   $[4..7] = \mathtt{baab}$, 
   $[5..8] = \mathtt{aaba}$, and 
   $[6..9] = \mathtt{abab}$ are SUSs for query position $6$,
   where the first SUS $[3..6]$ is obtained by extending
   the right-end of MUS $[3..4]$ up to position $6$,
   the second SUS $[4..7]$ and the third $[5..8]$ are MUSs of $S$,
   and the fourth SUS $[6..9]$ is obtained by extending
   the left-end of MUS $[8..11]$ up to position $6$.
 }
\label{fig:MUSandSUS}
\end{figure}

Hu et al.~\cite{Hu} showed that
it is possible to preprocess a given string $S$ of length $n$
in $O(n)$ time
so that later, we can return all SUSs that contain
a query interval $[s,t]$ in $O(k)$ time,
where $k$ is the number of such SUSs.

As is shown in Lemma~\ref{lem:sizeof_mus},
the number of MUSs in any string $S$ of length $n$ is bounded by $n$.
In this paper, we show that the number of point SUSs in $S$
is less than $1.5n$, more precisely, $|\mathcal{PS}_S| \leq (3n-1)/2$.
We will do so by first showing two different bounds on $|\mathcal{PS}_S|$ 
in terms of the number $|\mathcal{M}_S|$ of MUSs in the string $S$,
and then merging these two results that lead to the claimed bound.
Moreover, this bound is indeed tight, namely,
we show a series of strings containing $(3n-1)/2$ SUSs.
In addition, 
we show that the number of non-trivial SUSs in $S$
is less than $2n$, namely, $|\mathcal{IS}_S| < 2n$.
We also prove that there exists 
a string of length $n$ which contains 
$(2-\varepsilon)n$ non-trivial SUSs 
for any small number $\varepsilon > 0$.

\section{Bounds on the number of point SUSs} \label{sec:bounds}

Here we show a tight bound
for the maximum number of point SUSs in a string.
In this section, whenever we speak of SUSs,
we mean point SUSs (those for the point SUS problem).

\subsection{Upperbound A} \label{subsec:uboundA}

In this subsection, we show our first upperbound 
on the number of SUSs in a string $S$.
In so doing, we define the subsets $\LS_S$, $\MS_S$, 
and $\RS_S$ of the set $\mathcal{PS}_S$ of all SUS of string $S$ by
\begin{align*}
\LS_S
 =&~
 \mathcal{PS}_S \cap \{[x,y] \not\in \mathcal{M}_S : x < \exists i\leq y~[i,y]\in\mathcal{M}_S\}, \\
\MS_S =&~ \mathcal{PS}_S \cap \mathcal{M}_S, \mbox{ and} \\
\RS_S
 =&~
 \mathcal{PS}_S \cap \{[x,y] \not\in \mathcal{M}_S : x \leq \exists j < y~[x,j]\in\mathcal{M}_S\}.
\end{align*}
Intuitively, $\LS_S$ is the set of SUSs of $S$
which are \emph{not} MUSs of $S$ and can be obtained by
extending the beginning positions of some MUSs 
to the left up to query positions,
$\MS_S$ is the set of SUSs of $S$
which are also MUSs of $S$,
and $\RS_S$ is the set of SUSs of $S$
which are \emph{not} MUSs of $S$ and can be obtained by
extending the ending positions of some MUSs 
to the right up to query positions.

It follows from their definitions that
$\LS_S \cap \MS_S = \phi$, 
$\MS_S \cap \RS_S = \phi$, 
$\RS_S \cap \LS_S = \phi$ 
and that
$\mathcal{PS}_S = 
\LS_S \cup \MS_S \cup \RS_S$.

Figure~\ref{fig:3TypeSUSs} in the next subsection shows 
examples of $\LS_S$, $\MS_S$, and $\RS_S$
for string $S = \mathtt{aabbaababaa}$.
Also compare it with Figure~\ref{fig:MUSandSUS} which shows 
$\mathcal{PS}_S$ for the same string $S$.

In the proof of the following theorem,
we will evaluate the sizes of these three sets 
$\LS_S$, $\MS_S$, and $\RS_S$ separately.

\begin{theorem} \label{thm:ubound1}
For any string $S$, 
$|\mathcal{PS}_S| \leq 2|S|-|\mathcal{M}_S|$.
\end{theorem}

\begin{proof}
Let $n = |S|$ and $m = |\mathcal{M}_S|$.
For any $1 \leq i \leq m$,
let $[b_i, e_i]$ denote the MUS of $S$
that has the $i$th smallest beginning position in $\mathcal{M}_S$.

It is clear that $|\MS_S| \leq m$.
Note that the inequality is due to that fact that 
some MUS may not be a point SUS for any position in $S$
(such a MUS is called \emph{meaningless} in the literature~\cite{Tsuruta}).

Next, we consider the size of $\RS_S$.
By definition, for any $[x, y] \in \RS_S$,
$x$ is equal to the beginning position of a MUS of $S$.
Therefore, we can bound $|\RS_S|$ by summing up
the number of SUSs that begin with $b_i$ for every $[b_i, e_i] \in \mathcal{M}_S$.
For any $1 \leq i \leq m-1$, consider two adjacent MUSs
$[b_i, e_i], [b_{i+1}, e_{i+1}] \in \mathcal{M}_S$.
Recall that $b_i < b_{i+1}$.
Then, for any $j \geq e_{i+1}$,
the interval $[b_i, j]$ contains both MUSs $[b_i, e_i]$ 
and $[b_{i+1}, e_{i+1}]$.
This implies that $[b_i, j] \not \in \mathcal{PS}_S$ (see Figure \ref{fig:RS}),
since otherwise both $[b_i, j]$ and $[b_{i+1}, j]$ are SUSs for position $j$,
a contradiction.
Thus, for any $[b_i, e_i] \in \mathcal{M}_S$ with $1 \leq i \leq m-1$,
the number of SUSs that begin with $b_i$ and belong to $\RS_S$
is at most $e_{i+1}-e_{i}-1$.
Also, the number of SUSs that begin with $b_m$ 
and belong to $\RS_S$ is at most $n-e_m$.
Consequently, we get
$|\RS_S| = \sum_{i=1}^{m-1}(e_{i+1}-e_{i}-1) + n-e_m = e_m - e_1 - (m-1) + n - e_m \leq n-m$.

A symmetric argument gives us the same bound for $|\LS_S|$,
namely, $|\LS_S| \leq n-m$.
Overall, we obtain
$|\mathcal{PS}_S| = |\LS_S| + |\MS_S| + |\RS_S| \leq 2(n-m) + m  = 2n-m$.
\end{proof}

\begin{figure}[tb] 
 \centerline{
  \includegraphics[scale=0.3]{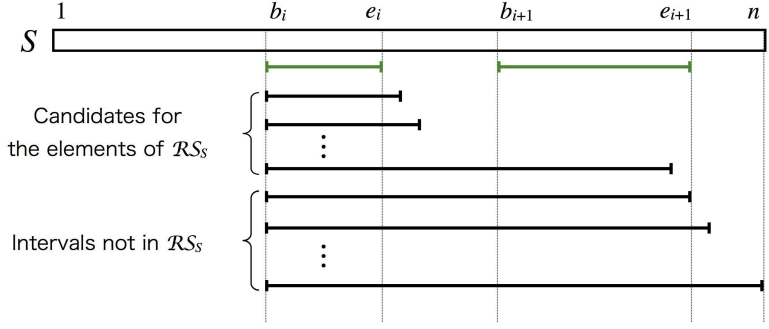}
 }
 \caption{
   Illustration for Theorem~\ref{thm:ubound1}.
   Consider two adjacent MUSs $[b_i, e_i]$ and $[b_{i+1}, e_{i+1}]$
   depicted as the two intervals on the top.
   For any $e_i < e < e_{i+1}$, $[b_i, e]$ can be 
   an element of $\RS_S$.
   On the other hand, for any $e' \geq e_{i+1}$,
   $[b_i, e']$ can never be an element of $\mathcal{PS}_{S}$
   since $[b_i, e']$ contains two distinct MUSs
   $[b_i, e_i]$ and $[b_{i}, e_{i+1}]$,
   and hence $[b_i, e']$ can never be an element of $\RS_S$ as well.
 }
\label{fig:RS}
\end{figure}

\subsection{Upperbound B} \label{subsec:uboundB}

In this subsection, we provide another upperbound
on the size of $\mathcal{PS}_S$.

\begin{theorem}\label{thm:ubound2}
For any string $S$,
$|\mathcal{PS}_S| \leq |S|+|\mathcal{M}_S|-1$.
\end{theorem}

In order to show Theorem~\ref{thm:ubound2},
we will use a function $f : \mathcal{PS}_S \rightarrow \{1, 2, \ldots, n\}$
and its inverse image
$f^{-1}: \{1, 2, \ldots, n\} \rightarrow 2^{\mathcal{PS}_S}$.
The next lemma is useful to define $f$ and $f^{-1}$.
\begin{lemma}\label{lem:propertyofS}
For any string $S$ and interval $[x, y]$ such that $1 \leq x \leq y \leq |S|$,
if $[x, y] \in \RS_S$ then $[x, y] \in \SUS_S(y)$,
and   
if $[x, y] \in \LS_S$ then $[x, y] \in \SUS_S(x)$.
\end{lemma}

\begin{proof}
We first prove the former case.
Assume on the contrary that
some $[x, y] \in \RS_S$ satisfies $[x, y] \not\in \SUS_S(y)$.
This implies that there exists a position $p$ in $S$ such that
$x \leq p < y$ and $[x, y] \in \SUS_S(p)$.
In addition, since $[x, y] \in \RS_S$,
there exists a position $q$
such that $x \leq q < y$ and $[x, q] \in \mathcal{M}_S$.
Let $z = \max\{p, q\}$.
Then, $S[x..z]$ is a unique substring of $S$
which is shorter than $S[x..y]$ and contains position $p$.
However, this contradicts that $S[x..y]$ is a SUS for position $p$.
Thus, if $[x, y] \in \RS_S$ then $[x, y] \in \SUS_S(y)$.
The latter case is symmetric and thus can be shown similarly.
\end{proof}

We are now ready to define $f$:
\[
f([x, y]) = 
\begin{cases}
    x & \mbox{if }[x,y] \in \LS_S \cup \MS_S,\\
    y & \mbox{if }[x,y] \in \RS_S.
\end{cases}
\]
Intuitively, the function $f$ charges a given interval $[x, y]$
to its beginning position $x$
if $[x, y]$ is an element of $\mathcal{M}_S \cap \mathcal{PS}_S$ or 
if $[x, y]$ is an element of $\SUS_S(p)$ for some query position $p$
which is obtained by extending the left-end of a MUS to the left
up to $p$.
On the other hand, it charges $[x, y]$ to
its ending position $y$ if the interval is an element of $\SUS_S(p)$
for some query position $p$
which is obtained by extending the right-end of a MUS to the right
up to $p$.
Figure \ref{fig:3TypeSUSs} shows examples for how
the function $f$ charges given interval $[x, y] \in \mathcal{PS}_S$.

We also define the inverse image $f^{-1}$ of $f$ as follows:
\[
f^{-1}(u) = \{[x,y]\in \mathcal{PS}_S~:~f([x,y]) = u\}.
\]
For positions $u$ for which there is no element $[x, y]$ in $\mathcal{PS}_S$
satisfying $f([x, y]) = u$, let $f^{-1}(u) = \emptyset$.
See also Figure \ref{fig:3TypeSUSs} for examples of $f^{-1}$.

By the definition of $f^{-1}$, it is clear that
$|\mathcal{PS}_S| = \sum^{|S|}_{u=1}|f^{-1}(u)|$.
Hence, in what follows we analyze $|f^{-1}(u)|$
for all positions $u$ in string $S$.

\begin{figure}[tb] 
 \centerline{
  \includegraphics[width=100mm]{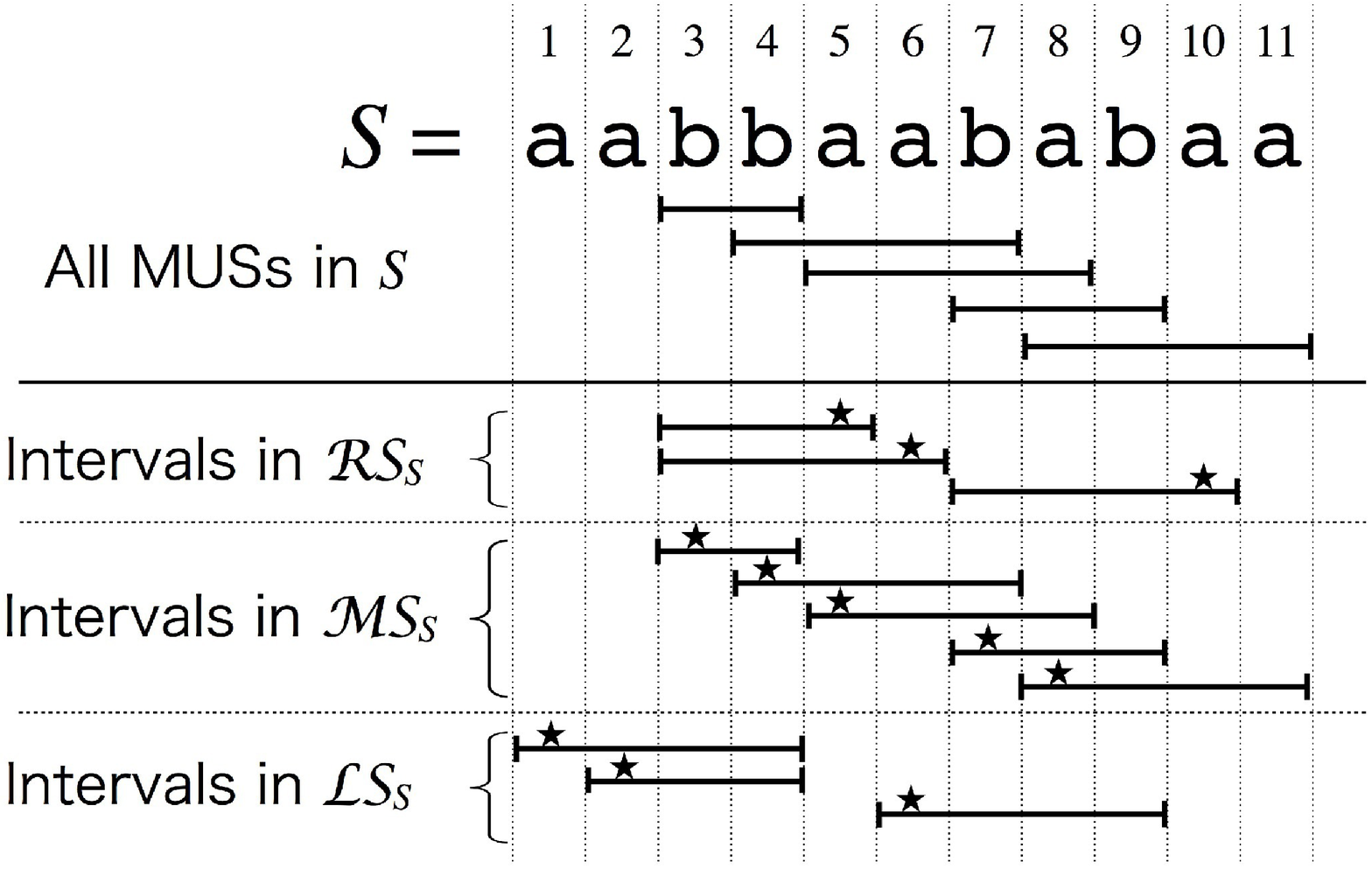}
 }
 \caption{
   Illustration for functions $f$ and $f^{-1}$
   of string $S = \mathtt{aabbaababaa}$.
   The upper part of this diagram shows all MUSs in $S$,
   and the lower part shows all SUSs for all positions in $S$.
   Each star shows the position to which the function $f$ maps
   the corresponding interval. 
   Here, $\RS_S = \{ [3,5], [3,6], [7,10] \}$, 
   $\MS_S = \{ [3,4], [4,7], [5,8], [7,9], [8,11]\}$, and 
   $\LS_S = \{ [1,4], [2,4], [6,10]\}$.
   Hence, we have
   $f([3,5]) = 5$, $f([3,6]) = 6$, $f([7,10]) = 10$, 
   $f([3,4]) = 3$, $f([4,7]) = 4$, $f([5,8]) = 5$, $f([7,9]) = 7$, $f([8,11]) = 8$, 
   $f([1,4]) = 1$, $f([2,4]) = 2$, and $f([6,10]) = 6$.
   For the inverse image, $f^{-1}$, we have $f^{-1}(1) = \{[1,4]\}$,
   $f^{-1}(2) = \{[2,4]\}$, $f^{-1}(3) = \{[3,4]\}$, $f^{-1}(4) = \{[4,7]\}$,
   $f^{-1}(5) = \{[3,5], [5,8]\}$, $f^{-1}(6) = \{[3,6], [6,10]\}$,
   $f^{-1}(7) = \{[7,9]\}$, $f^{-1}(8) = \{[8,11]\}$, 
   $f^{-1}(9) = f^{-1}(11) = \emptyset$, and $f^{-1}(10) = \{[7,10]\}$.
 }
\label{fig:3TypeSUSs}
\end{figure}

\begin{lemma}\label{lem:sizeof_finv}
For any string and position $1 \leq u \leq |S|$,
$|f^{-1}(u)| \leq 2$.
\end{lemma}

\begin{proof}
Assume on the contrary that $|f^{-1}(u)| \geq 3$ for some position $u$ in $S$.
Let $[x_1, y_1]$, $[x_2, y_2]$ be any distinct elements of $f^{-1}(u)$.
We firstly consider the following cases.
\begin{enumerate}
\item[(1)] Case where $[x_1, y_1], [x_2, y_2] \in \LS_S$:
  It follows from the definition of $f^{-1}$ that $f([x_1, y_1]) = f([x_2, y_2]) = u$,
  and it follows from the the definition of $f$ that $x_1 = x_2 = u$.
  Since $[x_1, y_1]$ and $[x_2, y_2]$ are distinct, $y_1 \neq y_2$.
  Assume w.l.o.g. that $y_1 < y_2$.
  Then, $[x_2, y_2] = [u, y_2]$ is a SUS for position $u$
  but it is longer than another SUS $[x_1, y_1] = [u, y_1]$ for position $u$,
  a contradiction.

\item[(2)] Case where $[x_1, y_1], [x_2, y_2] \in \MS_S$:
  It follows from the definition of $f^{-1}$ that $f([x_1, y_1]) = f([x_2, y_2]) = u$,
  and it follows from the definition of $f$ that $x_1 = x_2 = u$.
  Since $[x_1, y_1]$ and $[x_2, y_2]$ are distinct, $y_1 \neq y_2$.
  Assume w.l.o.g. that $y_1 < y_2$.
  Then, $[x_2, y_2] = [u, y_2]$ is a MUS,
  but it contains another MUS $[x_1, y_1] = [u, y_1]$,
  a contradiction.

\item[(3)] Case where $[x_1, y_1], [x_2, y_2] \in \RS_S$:
  This is symmetric to Case (1) and thus we can obtain
  a contradiction in a similar way.
\end{enumerate}

Hence, none of the above three cases is possible,
and thus the remaining possibility is the case
where $|f^{-1}(u)| = 3$ and each element of $f^{-1}(u)$
belongs to a different subset of $\mathcal{PS}_S$,
namely, $f^{-1}(u) = \{[x_1, y_1], [x_2, y_2], [x_3, y_3]\}$
for some $[x_1, y_1] \in \LS_S$,
$[x_2, y_2] \in \mathcal{MS}_S$, and $[x_3, y_3] \in \RS_S$.
It follows from the definition of $f^{-1}$ that $f([x_1, y_1]) = f([x_2, y_2]) = u$,
and it follows from the definition of $f$ that $x_1 = x_2 = u$.
Since $[x_1, y_1]$ and $[x_2, y_2]$ are distinct, $y_1 \neq y_2$.
There are two sub-cases.
\begin{enumerate}
  \item[(i)]
  If $y_1 < y_2$, then a MUS $[x_2, y_2] = [u, y_2]$ contains
  a shorter SUS $[x_1, y_1] = [u, y_1]$ for position $u$, a contradiction.

  \item[(ii)]
  If $y_1 > y_2$, then a SUS $[x_1, y_1] = [u, y_1]$ for position $u$
  contains a shorter MUS $[x_2, y_2] = [u, y_2]$, a contradiction.
\end{enumerate}
Hence, neither of the sub-cases is possible.

Overall, we conclude that $|f^{-1}(u)| \leq 2$.
\end{proof}

By Lemma~\ref{lem:sizeof_finv},
for any position $u$ in string $S$ we have $|f^{-1}(u)| \leq 2$.
Now let us consider any position $u$ for which $|f^{-1}(u)| = 2$.
We have the next lemma.
\begin{lemma}\label{lem:sizeof_finv2}
  For any position $u$ in string $S$ for which $|f^{-1}(u)| = 2$,
  let $f^{-1}(u) = \{[x_1, y_1], [x_2, y_2]\}$ and assume
  w.l.o.g. that $x_1 \leq x_2$.
  Then, $x_1 \neq x_2$, $[x_1, y_1] \in \RS_S$ and
  $[x_2, y_2] \in \LS_S \cup \MS_S$.
\end{lemma}

\begin{proof}
Suppose $x_1 = x_2$ and assume w.l.o.g. that $y_1 < y_2$.
Then, from the definition of $f$, we have that
($x_1 = u$ or $y_1 = u$) and ($x_2 = u$ or $y_2 = u$)
and thus $x_1 = x_2 = u$.
Since $[x_2,y_2]\in f^{-1}(u)$ is not a MUS since it includes $[x_1,y_1]$, it must be that $[x_2,y_2] \in \SUS_S(u)$.
This is a contradiction, because there exists a shorter unique substring $[x_1, y_1]$ that contains $u$.
Thus we have $x_1\neq x_2$.
Assume on the contrary that 
$[x_1, y_1] \in \LS_S \cup \MS_S$.
Then, it follows from the definition of $f$ that $f([x_1, y_1]) = x_1$.
In addition, since $[x_1, y_1] \in f^{-1}(u)$, we have $u  = x_1$.
This implies that $u = x_1 < x_2$,
but it contradicts that $[x_2, y_2] \in f^{-1}(u)$.
Thus, $[x_1, y_1] \not\in \LS_S \cup \MS_S$,
namely, $[x_1, y_1] \in \RS_S$.
Now, it follows from the arguments in the proof of Lemma~\ref{lem:sizeof_finv}
that $[x_2, y_2] \not\in \RS_S$,
and hence $[x_2, y_2] \in \MS_S \cup \LS_S$.
\end{proof}

Let $m = |\mathcal{M}_S|$,
and $\mathcal{M}_S = \{[b_1, e_1], \ldots, [b_m, e_m]\}$.
The next corollary immediately follows from Lemmas~\ref{lem:propertyofS}
and \ref{lem:sizeof_finv2}.
\begin{corollary}\label{col:finv2}
For any position $u$ in string $S$
with $|f^{-1}(u)| = 2$,
there exist two integers $1 \leq i < j \leq m$
such that $\SUS_S(u) = \{[b_i, u], [u, e_j]\}$.
\end{corollary}

For any position $u$ in string $S$ before $b_1$ or after $b_m$,
we have the next lemma.

\begin{lemma}\label{lem:endSUSnum}
For any position $u$ in string $S$ s.t.
$1 \leq u \leq b_1$ or $b_m < u \leq n$, $|f^{-1}(u)| \leq 1$.
\end{lemma}

\begin{proof}
Assume on the contrary that 
$|f^{-1}(u)| = 2$ for some $1 \leq u \leq b_1$.
By Lemma~\ref{lem:sizeof_finv2},
there exists $[x, y] \in f^{-1}(u)$ such that $[x, y] \in \RS_S$.
By the definitions of $f$ and $f^{-1}$, we have $y = u$.
Also, by the definition of $\RS_S$,
there exists a position $e < y$ in $S$ such that $[x,e] \in \mathcal{M}_S$.
Now we have $x \leq e < y = u \leq b_1$,
however, this contradicts that $b_1$ is the beginning position of 
the first (leftmost) MUS in $\mathcal{M}_S$.
Thus $|f^{-1}(u)| \leq 1$ for any $1 \leq u \leq b_1$.

Assume on the contrary that 
$|f^{-1}(u)| = 2$ for some $b_m < u \leq n$.
By Lemma~\ref{lem:sizeof_finv2},
there exists $[x', y'] \in f^{-1}(u)$ such that 
$[x', y'] \in \MS_S \cup \LS_S$.
By the definition of $f$ and $f^{-1}$,
we have $x' = u$.
There are two cases to consider:
\begin{itemize}
 \item If $[x', y'] \in \MS_S$,
       then $[x', y'] \in \mathcal{M}_S$.
       Thus $x' = u > b_m$ is the beginning position of a MUS in $\mathcal{M}_S$,
       however, this contradicts that $b_m$ is the beginning position of the 
       last (rightmost) MUS in $\mathcal{M}_S$.

 \item If $[x', y'] \in \LS_S$,
       then by the definition of $\LS_S$
       there exists a position $b > x'$ such that $[b,y'] \in \mathcal{M}_S$.
       Now we have $b > x' = u > b_m$,
       however, this contradicts that $b_m$ is the beginning position of the 
       last (rightmost) MUS in $\mathcal{M}_S$.
\end{itemize}
Consequently, $|f^{-1}(u)| \leq 1$ for any $b_m < u \leq n$.
\end{proof}

\begin{lemma}\label{lem:sizeof_U}
  For any non-empty string $S$,
  let $U = \{ u:|f^{-1}(u)| = 2\}$.
  Then, $|U| \leq |\mathcal{M}_S| - 1$.
\end{lemma}

\begin{proof}
Let $n = |S|$ and $m = |\mathcal{M}_S|$.
Recall that for any $1 \leq i \leq m$,
$[b_i, e_i]$ denotes the $i$th element of $\mathcal{M}_S$.

Let $B = \{ b_i~:~1 \leq i \leq m-1\}$.
We define function $g : U \rightarrow B$ as
$g(u) = \max\{ b < u~:~b \in B \}$.
By the definition of $U$ and Lemma~\ref{lem:endSUSnum},
any position $u \in U$ satisfies $b_1 < u \leq b_m$.
Therefore, $g(u)$ is well-defined for any position $u \in U$,
and $g(u)$ returns the predecessor of $u$ in the set $B$.
It is clear that $|B| = m-1$.
Thus, if $g$ is an injection,
then we immediately obtain the claimed bound $|U| \leq |B| = m-1$.

In what follows, we show that $g$ is indeed an injection.
Assume on the contrary that $g$ is not an injection.
Let $u_1$ and $u_2$ be elements in $U$ such that $u_1 < u_2$
and $g(u_1) = g(u_2)$.
Let $b_i \in B$ such that $b_i = g(u_1) = g(u_2)$.
Then, by the definition of $g$, we have $b_i < u_1 < u_2 \leq b_{i+1}$.
See Figure \ref{fig:injective} for illustration.

Let $l_1$ and $l_2$ be the lengths of the SUSs for positions $u_1$ and $u_2$,
respectively.
Since $|f^{-1}(u_2)| = 2$,
it follows from Corollary~\ref{col:finv2} that
there exists $b_k \in B$ such that $b_k \leq b_i$ and
$\SUS_S(u_2) = \{[b_k, u_2], [u_2, e_{i+1}]\}$.
This implies $l_2 = u_2-b_k+1 = e_{i+1}-u_2+1$.
On the other hand, since $|f^{-1}(u_1)| = 2$,
it follows from Corollary~\ref{col:finv2} that
$[u_1, e_{i+1}] \in \SUS_S(u_1)$,
which implies $l_1 = e_{i+1}-u_1+1$.
Since $u_1 < u_2$, we have $l_1 > l_2$.

Now focus on a SUS $[b_k, u_2]$ for position $u_2$.
Since $b_k \leq b_i < u_1 < u_2$, $[b_k, u_2]$ contains $u_1$.
However, $[b_k, u_2]$ is a SUS for position $u_2$
and is of length $l_2 < l_1$.
This contradicts that $[u_1, e_{i+1}]$ of length $l_1$ is each SUS
for position $u_1$.
Hence $g$ is an injection.
\end{proof}

\begin{figure}[tb] 
 \centerline{
  \includegraphics[scale=0.3]{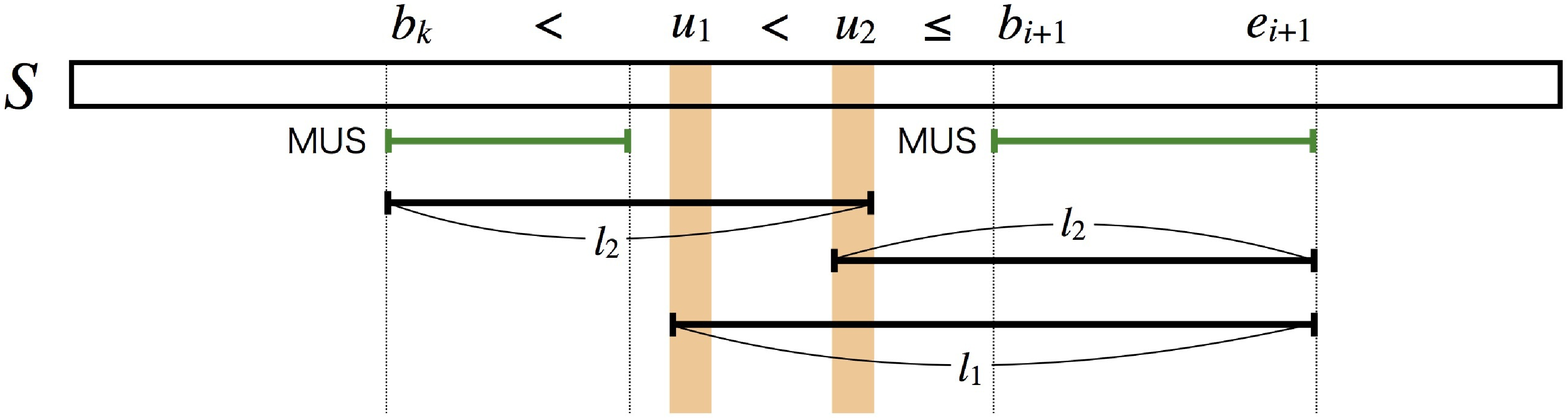}
 }
 \caption{
   Illustration for Lemma~\ref{lem:sizeof_U}.
   The two intervals show two MUSs $[b_k, e_k]$, $[b_{i+1}, e_{i+1}]$
   $\in \mathcal{M}_S$,
   where $b_k \leq b_i$.
   Both $[b_k, u_2]$ and $[u_2, b_{i+1}]$ are SUSs for position $u_2$,
   and $[u_1, e_{i+1}]$ is a SUS for position $u_1$.
   Since $u_1 < u_2$,
   it holds that $l_1 > l_2$,
   where $l_1$ and $l_2$ are the lengths of SUSs for positions $u_1$ and $u_2$,
   respectively.
   Then, the interval $[b_k, u_2]$ of length $l_2$ contains position $u_1$
   and $S[b_k..u_2]$ is a unique substring of $S$.
   However, this contradicts that $l_1$ is the length of each SUS for position $u_1$.
 }
\label{fig:injective}
\end{figure}

We are ready to prove the main result of this subsection, Theorem~\ref{thm:ubound2}.
\begin{proof}
Let $n = |S|$, $m = |\mathcal{M}_S|$,
$U = \{u : |f^{-1} (u)| = 2\}$, and 
$V = \{1, \cdots , n\} \setminus U$.
It is clear that $|U| + |V| = n$.
By Lemma~\ref{lem:sizeof_finv}, $V = \{u \ : \ |f^{-1} (u)| \leq 1\}$.
Also, by Lemma~\ref{lem:sizeof_U}, $|U| \leq m-1$.
Recall that $|\mathcal{PS}_S| = \sum^{n}_{u=1}|f^{-1}(u)|$.
Putting all together, we obtain
$|\mathcal{PS}_S| = \sum^{n}_{u=1}|f^{-1}(u)| 
\leq |V| + 2|U| = n + |U| \leq  n + m - 1$. 
\end{proof}

\subsection{Matching upper and lower bounds} \label{subsec:tight}

We are ready to show the main result of this paper.

\begin{theorem} \label{thm:ubound_is_tight}
  For any non-empty string $S$,
  $|\mathcal{PS}_S| \leq (3|S|-1)/2$.
  This bound is tight, namely, 
  for any odd $n \geq 5$
  there exists a string $T$ of length $n$ s.t.
  $|\mathcal{PS}_T| = (3n-1)/2$.
\end{theorem}

\begin{proof}
  By Theorem~\ref{thm:ubound1},
  we have $|\mathcal{M}_S| \leq 2|S| - |\mathcal{PS}_S|$.
  Also, by Theorem~\ref{thm:ubound2},
  we have $|\mathcal{PS}_S|-|S|+1 \leq |\mathcal{M}_S|$.
  Thus $|\mathcal{PS}_S|-|S|+1 \leq 2|S| - |\mathcal{PS}_S|$,
  which immediately leads to the claimed bound $|\mathcal{PS}_S| \leq (3|S|-1)/2$.
  
  We show that the above upperbound is indeed tight.
  For any odd number $n = 2k-1 \geq 5$,
  consider string
  $T = a_1 x a_2 x \cdots a_{k-1} x a_k$,
  where $a_1, \ldots, a_k, x \in \Sigma$,
  $a_i \neq a_j$ for all $1 \leq i \neq j \leq k$,
  and $x \neq a_i$ for all $1 \leq i \leq k$.
  For any $1 \leq i \leq k$,
  $T[2i-1] = a_i$ is a unique substring of $T$,
  and thus $[2i-1, 2i-1] \in \SUS_T(2i-1)$.
  Also, for any $1 \leq i \leq k-1$,
  $T[2i] = x$ is a repeating substring of $T$
  while $T[2i-1..2i] = a_ix$ and $T[2i..2i+1] = xa_{i+1}$
  are unique substrings of $T$.
  This implies that $[2i-1, 2i], [2i, 2i+1] \in \SUS_T(2i)$.
  Hence, we have
  $|\mathcal{PS}_T| = k + 2 (k - 1) = 3k - 2 = 3(n+1)/2 - 2 = (3n - 1) / 2$.
\end{proof}

\subsection{Lower bound for fixed-size alphabet}
\label{subsec:binary}

The lowerbound of Theorem~\ref{thm:ubound_is_tight}
is due to a series of strings over an alphabet of unbounded size.
In this subsection, we fix the alphabet size $\sigma$
and present a series of strings that contain many point SUSs.

\begin{theorem} \label{theo:constant}
Let $n \geq 2$ and $2 \leq \sigma \leq (n+3)/2$.
There exists a string $T$ of length $n$
over an alphabet of size $\sigma$ such that $|\mathcal{PS}_T| = n+\sigma-2$.
\end{theorem}

\begin{proof}
Let $\Sigma = \{a_1, \cdots, a_{\sigma-1}, x\}$ 
and $T = a_1 x a_2 x \cdots a_{\sigma-1} x^{n-2\sigma+3}$.
  For any $1 \leq i \leq \sigma-1$, 
  $T[2i-1] = a_i$ is a unique substring of $T$, 
  and thus $[2i-1,2i-1] \in \SUS_T(2i-1)$.
  For any $1 \leq j \leq \sigma-2$, 
  $T[2j] = x$ is a repeating substring of $T$
  while $T[2j-1..2j] = a_j x$ and $T[2j..2j + 1] = x a_{j+1}$
  are unique substrings of $T$. 
  This implies that $[2j-1,2j], [2j,2j+1] \in \SUS_T(2j)$.
  For any $2\sigma-2 \leq k \leq n-1$, $T[2\sigma-2..k] = x^{k-2\sigma+3}$
  is a repeating substring of $T$
  while $T[2\sigma-1..k] = a_{\sigma-1} x^{k-2\sigma+3}$ is
  a unique substrings of $T$.
  This implies that $[2\sigma-1,k] \in \SUS_T(k)$.
  Also, $T[2\sigma-1..n] = x^{n-2\sigma+2}$
  is a repeating substring of $T$
  and $T[2\sigma-2..n] = x^{n-2\sigma+3}$ is a unique substring of $T$,
  and thus $[2\sigma-2..n] \in \SUS_T(n)$.
Summing up all the point SUSs above,
we obtain $|\mathcal{PS}_T| = \sigma-1 + 2(\sigma-2) + n-2\sigma+2 + 1 
= n+\sigma-2$.
\end{proof}

\section{Bounds on the number of interval SUSs} \label{sec:interval}
In this section, 
we show the tight bound 
for the maximum number of 
non-trivial interval SUSs $\mathcal{IS}_S$ of a string $S$.
The following upper bound for $|\mathcal{IS}_S|$ can be obtained 
in an analogous way to Theorem~\ref{thm:ubound1}.
\begin{lemma} \label{lem:intervalubound}
For any non-empty string $S$, 
$|\mathcal{IS}_S| \leq 2|S|-|\mathcal{M}_S|$.
\end{lemma}
We also have the following lower bound for $|\mathcal{IS}_S|$.
\begin{lemma} \label{lem:intervallbound}
For any $\varepsilon > 0$, 
there exists a string $T$ of length $n$ 
such that 
$|\mathcal{IS}_T| > (2-\varepsilon)n$.
\end{lemma}

\begin{proof}
Let $x = \lceil 3/(2\varepsilon) \rceil$, 
$T = c_1 a^x c_2 a^x c_3$ and $n = |T| = 2x+3$.
Clearly, $c_1, c_2$ and $c_3$ are MUSs of $T$ 
and are in $\mathcal{IS}_T$. 
For all $2\leq i \leq x+1$, 
$T[1..i]$ and $T[i..x+2]$ are unique substrings of $T$, 
and $T[2..i]$ and $T[i..x+1]$ are repeating substrings of $T$.
This implies $T[1..i] \in \SUS_S([2, i])$ 
and $T[i..x+2] \in \SUS_S([i, x+1])$.
Similarly, for all $x+3\leq j \leq 2x+2$, 
$T[x+2..j] \in \SUS_S([x+3, j])$ 
and $T[j..2x+3]\in \SUS_S([j, 2x+2])$.
Then, we have $|\mathcal{IS}_T| = 4x+3$.
Hence,
$|\mathcal{IS}_T| - (2-\varepsilon)n = 4x+3-(2-\varepsilon)(2x+3) 
= 2\varepsilon x+3\varepsilon-3 
= 2\varepsilon \lceil 3/(2\varepsilon) \rceil+3\varepsilon-3 
\geq 3\varepsilon > 0$.
\end{proof}

There exists a string for which the bounds of Lemma~\ref{lem:intervalubound} and Lemma~\ref{lem:intervallbound} almost match, namely:
\begin{theorem} \label{thm:intervalbound}
For any $\varepsilon > 0$,
there is a string $T$ such that $(2|T|-|\mathcal{M}_T|)-(2-\varepsilon)|T| \leq 5\varepsilon$.
\end{theorem}

\begin{proof}
For any $\varepsilon > 0$, consider the string $T$ of Lemma~\ref{lem:intervallbound}.
We remark that $T$ contains $3$ MUSs, namely, $|\mathcal{M}_T| = 3$.
Hence, we obtain
$(2|T|-|\mathcal{M}_T|)-(2-\varepsilon)|T| = \varepsilon|T|-|\mathcal{M}_T| = \varepsilon|T|-3 = \varepsilon(2\lceil 3/(2\varepsilon) \rceil+3)-3 = 2\varepsilon\lceil 3/(2\varepsilon) \rceil+3\varepsilon-3 \leq 2\varepsilon(3/(2\varepsilon) +1)+3\varepsilon-3 = 5\varepsilon \rightarrow 0\, (\varepsilon\rightarrow 0)$.
\end{proof}

\section{Conclusions and open question} \label{sec:conc}
In this paper, we presented matching upper
and lower bounds for the maximum number of SUSs
for the point SUS problem.
Namely, we proved that any string of length $n$
can contain at most $(3n-1)/2$ SUSs for
the point SUS problem,
and showed that this bound is tight by giving
a string of length $n$ containing $(3n-1)/2$ SUSs.
For a fixed alphabet size $\sigma$,
we also presented a string of length $n$
containing $n+\sigma-2$ SUSs.
Moreover, we showed that any string of length $n$
which contains $m$ MUSs can have
at most $2n-m$ non-trivial interval SUSs,
and that for any $\varepsilon > 0$ there is a string of length $n$
which contains $(2-\varepsilon)n$ non-trivial interval SUSs.

An interesting open question is to show a
non-trivial upper bound of the maximum number of point SUSs
for a fixed alphabet size $\sigma$.
We conjecture that the tight upper bound matches our lower bound $n+\sigma-2$.

\bibliographystyle{abbrv}
\bibliography{ref}

\end{document}